\documentclass{article}
\usepackage[totalwidth=13.0cm,totalheight=20.0cm]{geometry}

\usepackage{latexsym,amsthm,amsmath,amssymb,url}

\newcommand{\red}{{\rm red}}
\newcommand{\blue}{{\rm blue}}

\newtheorem{lemma}{Lemma}

\newtheorem{definition}{Definition}
\newtheorem{theorem}{Theorem}

\newtheorem{rrule}{Reduction Rule}

\newcommand{\desc}{\psi}

\newenvironment{parnamedefn}[4]{
\par\addvspace{0.4\baselineskip}\fbox{%
\begin{minipage}[t]{0.9\linewidth}%
\begin{tabular}{p{18mm}p{115mm}}
    \multicolumn{2}{l}{{{#1}}} \\
        \textsl{Input:} & {#2} \\ 
        \textsl{Parameter:} & {#3} \\ 
        \textsl{Question:} & {#4} \\
   \end{tabular}
\end{minipage}}\par\addvspace{0.4\baselineskip}
}

\begin{document}

\title{Parameterized Algorithms for Load Coloring Problem}

\author{
Gregory Gutin and Mark Jones\\
\small Department of Computer Science\\[-3pt]
\small  Royal Holloway, University of London\\[-3pt]
\small Egham, Surrey TW20 0EX, UK\\[-3pt]
\small \texttt{gutin|markj@cs.rhul.ac.uk}
}
\date{}
\maketitle

\newenvironment{compress}{\baselineskip=10pt}{\par}
\begin{abstract}
\noindent
One way to state the Load Coloring Problem (LCP) is as follows.
Let  $G=(V,E)$ be graph and let  $f:V\rightarrow \{\red, \blue\}$ be a 2-coloring. An edge $e\in E$ is called red (blue) if both end-vertices of $e$ are red (blue).
For a 2-coloring $f$, let $r'_f$ and $b'_f$ be the number of red and blue edges and let $\mu_f(G)=\min\{r'_f,b'_f\}$. Let $\mu(G)$ be the maximum of $\mu_f(G)$ over all 2-colorings. 

We introduce the parameterized problem $k$-LCP of deciding whether $\mu(G)\ge k$, where $k$ is the parameter.
We prove that this problem admits a kernel with at most $7k$. Ahuja et al. (2007) proved that one can find an optimal 2-coloring on trees in polynomial time. We generalize this by showing that an optimal 2-coloring on  graphs with tree decomposition of width $t$ can be found in time 
$O^*(2^t)$. We also show that either $G$ is a Yes-instance of $k$-LCP or the treewidth of $G$ is at most $2k$. Thus, $k$-LCP can be solved in time $O^*(4^k).$
\end{abstract}

\section{Introduction}\label{section:intro}

For a graph $G=(V,E)$ with $n$ vertices, $m$ edges and maximum vertex degree $\Delta$, the load distribution of a 2-coloring $f:V\rightarrow \{\red, \blue\}$ is a pair $(r_f,b_f),$ where $r_f$ is the number of edges with at least one end-vertex colored red and $b_f$ is the number of edges with at least one end-vertex colored blue. We wish to find a coloring $f$ such that the function $\lambda_f(G):=\max\{r_f,b_f\}$ is minimized. We will denote this minimum by $\lambda(G)$ and call this problem {\sc Load Coloring Problem} (LCP). 
The LCP arises in Wavelength Division Multiplexing, the technology used for constructing optical communication networks
\cite{AhBaDoPrSr,WDMbook}. 
Ahuja et al. \cite{AhBaDoPrSr} proved that the problem is NP-hard and gave a polynomial time algorithm for optimal colorings of trees. For graphs $G$ with genus $g>0$, Ahuja et al. \cite{AhBaDoPrSr} showed that a 2-coloring $f$ such that $\lambda_f(G)\le \lambda(G)(1+o(1))$ can be computed in $O(n+g\log n)$-time, if the maximum degree satisfies $\Delta=o(\frac{m^2}{ng})$ and an embedding is given.

For a 2-coloring $f:V\rightarrow \{\red, \blue\}$, let $r'_f$ and $b'_f$ be the number of edges whose end-vertices are both red and blue, respectively (we call such edges {\em red} and {\em blue}, respectively). Let $\mu_f(G):=\min\{r'_f,b'_f\}$ and let $\mu(G)$ be the maximum of $\mu_f(G)$ over all 2-colorings of $V$. It is not hard to see (and it is proved in Remark 1.1 of \cite{AhBaDoPrSr}) that $\lambda(G)=m-\mu(G)$ and so the LCP is equivalent to maximizing $\mu_f(G)$ over all 2-colorings of $V$.

In this paper we introduce and study the following parameterization of LCP.

\begin{parnamedefn}%
   {{\sc $k$-Load Coloring Problem} ($k$-LCP)}%
   {A graph $G=(V,E)$ and an integer $k$.}%
   {$k$}%
   {Is $\mu(G)\ge k$ ? (Equivalently, is $\lambda(G)\le m-k$?)}%
 \end{parnamedefn}
 
We provide basics on parameterized complexity and tree decompositions of graphs in the next section. 
In Section  \ref{sec:ker},  we show that $k$-LCP admits a kernel with at most $7k$ vertices. Interestingly, to achieve this linear bound, only two simple reduction rules are used.
In Section \ref{sec:tw}, we generalise the result of Ahuja et al. \cite{AhBaDoPrSr} on trees by showing that an optimal 2-coloring for graphs with tree decomposition of width $t$ can be obtained in time
$2^tn^{O(1)}.$ We also show that either $G$ is a Yes-instance of $k$-LCP or the treewidth of $G$ is at most $2k$. As a result, $k$-LCP can be solved in time $4^kn^{O(1)}.$
We conclude the paper in Section \ref{sec:concl} by stating some open problems.

\section{Basics on Fixed-Parameter Tractability, Kernelization and Tree Decompositions}\label{sec:fpt}

A \emph{parameterized problem} is a subset $L\subseteq \Sigma^* \times
\mathbb{N}$ over a finite alphabet $\Sigma$. $L$ is
\emph{fixed-parameter tractable} if the membership of an instance
$(x,k)$ in $\Sigma^* \times \mathbb{N}$ can be decided in time
$f(k)|x|^{O(1)},$ where $f$ is a function of the
parameter $k$ only. 
It is customary in parameterized algorithms to often write only the exponential part of $f(k)$: $O^*(t(k)):=O(t(k)(kn)^{O(1)})$.

Given a parameterized problem $L$,
a \emph{kernelization of $L$} is a polynomial-time
algorithm that maps an instance $(x,k)$ to an instance $(x',k')$ (the
\emph{kernel}) such that (i)~$(x,k)\in L$ if and only if
$(x',k')\in L$, (ii)~ $k'\leq g(k)$, and (iii)~$|x'|\leq g(k)$ for some
function $g$. The function $g(k)$ is called the {\em size} of the kernel.

It is well-known that
a parameterized problem $L$ is fixed-parameter
tractable if and only if it is decidable and admits a
kernelization. Due to applications, low degree polynomial size kernels are of main interest. Unfortunately, many fixed-parameter tractable problems do not have
kernels of polynomial size unless the polynomial hierarchy collapses to the third level, see, e.g., \cite{BodlaenderEtAl2009a,BoJaKr2011,BodlaenderEtAl2009}.
For further background and terminology on parameterized complexity we
refer the reader to the monographs~\cite{DowneyFellows99,FlumGrohe06,Niedermeier06}.

\begin{definition}\label{TDdef}
A \emph{tree decomposition} of a graph $G = (V,E)$ is a pair $(\mathcal{X},\mathcal{T})$, where $\mathcal{T} = (I,F)$ is a tree and $\mathcal{X} = \{ X_i: i \in I\}$ is a collection of subsets of $V$ called \emph{bags}, such that:
\begin{enumerate}
 \item $\bigcup_{i \in I} X_i  = V$;
 \item For every edge $xy \in E$, there exists $i \in I$ such that $\{x,y\} \subseteq X_i$;
 \item For every $x \in V$, the set $\{i: x \in X_i\}$ induces a connected subgraph of $\mathcal{T}$.
\end{enumerate}
The \emph{width} of $(\mathcal{T}, \mathcal{X})$ is $\max_{i \in I}|X_i|-1$. The \emph{treewidth} of a graph $G$ is the minimum width of all tree decompositions of $G$.
\end{definition}

To distinguish between vertices of $G$ and $\cal T$, we call vertices of $\cal T$ {\em nodes}.
We will often speak of a bag $X_i$ interchangeably with the node $i$ to which it corresponds in $\mathcal{T}$. Thus, for example, we might say two bags are {\em neighbors} if they correspond to nodes in $\mathcal{T}$ which are neighbors.
We define the \emph{descendants} of a bag $X_i$ as follows: every child of $X_i$ is a descendant of $X_i$, and every child of a descendant of $X_i$ is a descendant of $X_i$.

\begin{definition}
A \emph{nice tree decomposition} of a graph $G=(V,E)$ is a tree decomposition $(\mathcal{X},\mathcal{T})$ such that ${\cal T}$ is a rooted tree, and each node $i$  falls under one of the following classes:
 \begin{itemize}
  \item {\bf $i$ is a Leaf node:} Then $i$ has no children;
  \item {\bf $i$ is an Introduce node:} Then $i$ has a single child $j$, and there exists a vertex $v \notin X_j$ such that $X_i = X_j \cup \{v\}$;
  \item {\bf $i$ is a Forget node:} Then $i$ has a single child $j$, and there exists a vertex $v \in X_j$ such that $X_i = X_j \setminus \{v\}$;
  \item {\bf $i$ is a Join node:} Then $i$ has two children $h$ and $j$, and $X_i=X_h=X_j$.
 \end{itemize}
\end{definition}

It is known that any tree decomposition of a graph can be transformed into a tree decomposition of the same width.

 \begin{lemma} \cite{Klo94} \label{lem:niceform}
Given a tree decomposition with $O(n)$ nodes of a graph $G$ with $n$ vertices, we can construct, in time $O(n)$, a nice tree decomposition of $G$ of the same width and with at most $4n$ nodes.
\end{lemma}

\section{Linear Kernel}\label{sec:ker}

For a vertex $v$ of a graph $G=(V,E)$ and set $X\subseteq V$, let $\deg_{X}(v)$ denote the number of neighbors of $v$ in $X$. 
If $X=V$, we will write $\deg(v)$ instead of $\deg_{V}(v)$.

\begin{lemma}\label{lem:l1} Let $G=(V,E)$ be a graph with no isolated vertices, with maximum degree
$\Delta \ge 2$ and let $|V|\ge 5k$. If $|V|\ge 4k+\Delta$, then $(G,k)$ is a Yes-instance of $k$-LCP.
\end{lemma}
\begin{proof}
Suppose that $|V|\ge 4k+\Delta$, but $(G,k)$ is a No-instance of $k$-LCP.

Let $M$ be a maximum matching in $G$ and let $Y$ be the set of vertices which are not end-vertices of
edges in $M.$ If $M$ has at least $2k$ edges, then we may color half of them blue and half
of them red, so we conclude that $|M| < 2k.$

For an edge $e = uv$ in $M$, let $\deg_Y(e) = \deg_Y(u) + \deg_Y(v),$ that is the number
of edges between a vertex in $Y$ and a vertex of $e.$

\vspace{3mm}

\noindent{\bf Claim} {\em For any $e$ in $M,$ $\deg_Y(e) \le \max\{\Delta-1,2\}.$}

\vspace{2mm}

\noindent{\em Proof of Claim:} Suppose that $\deg_Y(e) \ge \Delta$ and let $e = uv.$ As $u$ and $v$ are adjacent,
$d_Y(u)$ and $d_Y(v)$ are each less than $\Delta$. But as $\deg_Y(u) + \deg_Y(v) = \deg_Y(e)
\ge \Delta,$ it follows that $\deg_Y(u) \ge 1$ and $\deg_Y(v) \ge 1.$ Then either $u$ and $v$
have only one neighbor in $Y,$ which is adjacent to both of them (in which
case $\deg_Y(e)=2$), or there exist vertices $x \neq y \in Y$ such that $x$ is adjacent
to $u$ and $y$ is adjacent to $v.$ Then $M$ is not a maximum matching, as $xuvy$ is
an augmenting path, which proves the claim.

\vspace{3mm}

Now let $M'$ be a subset of edges of $M$ such that 

\begin{equation}\label{eq1}
\sum_{e' \in M'} \deg_Y(e') \ge k - |M'|,
\end{equation}

and 

\begin{equation}\label{eq2}
\left[\sum_{e' \in M'}\deg_Y(e')\right] - \deg_Y(e) < k - |M'|, \text{ for any }e \in M'.
\end{equation}

To see that $M'$ exists observe first that $M'=M$ satisfies \eqref{eq1}. Indeed, suppose it is not true.
Then  $|V| < |V(M)| + k - |M| = k +|M| < 5k,$ a contradiction with our assumption that $|V|\ge 5k$.
Now let $M'$ be the minimal subset of $M$ that satisfies \eqref{eq1}, and observe that by minimality $M'$ also satisfies \eqref{eq2}.

Observe that $|M'|\le k.$
Then by the Claim, we have that $\sum_{e' \in M'} \deg_Y(e') \le k  - |M'|+ \Delta.$
Color $M'$ and all neighbors of $M'$ in $Y$ red, and note that there are at
most $k - |M'| + \Delta + 2|M'| = k +|M'| + \Delta \le 2k+\Delta$ such vertices. The number
of red edges is at least $k - |M'| + |M'| = k.$

Color the remaining vertices of $G$ blue. By assumption there are at least
$4k+\Delta - 2k - \Delta \ge 2k$ such vertices.
As $G$ contains no isolated vertices and $M$ is a maximum matching, the blue
vertices in $Y$ have neighbors in the vertices of $M \setminus M'.$ Thus, every blue
vertex has a blue neighbor. It follows that there are at least $2k/2 = k$
blue edges. Thus, $(G,k)$ is a Yes-instance of $k$-LCP.\end{proof}

We will use the following reduction rules for a graph $G$.

\begin{rrule}\label{r1} Delete isolated vertices.\end{rrule}

\begin{rrule}\label{r2} If there exists a vertex $x$ and set of vertices $S$ such
that $|S|>k$ and every $s \in  S$ has $x$ as its only neighbor, delete a vertex
from $S.$\end{rrule}

\begin{theorem}\label{thm:kernel}
The problem $k$-LCP has a kernel with at most $7k$ vertices.
\end{theorem}
\begin{proof}
Assume that $G$ is a graph reduced by Rules \ref{r1} and \ref{r2}. Assume also that $G$ is a No-instance. We will prove that $G$ has at most $7k$ vertices.

\vspace{3mm}

\noindent{\bf Claim A.} {\em There is no pair $x,y$ of distinct vertices such that $\deg(x)> 2k$ and $\deg(y) > k.$}

\vspace{2mm}

\noindent{\em Proof of Claim:}  Suppose such a pair $x,y$ exists. Color $y$ and $k$ of its neighbors, not including $x$, red. This leaves
$x$ and at least $k$ of its neighbors uncolored. Color $x$ and $k$ of its
neighbors blue. 

\vspace{3mm}

Construct $G'=(V',E')$ as follows. Let $x$ be a vertex in $G$ of maximum degree. Let $S$
be the vertices of $G$ whose only neighbor is $x.$ Then let $G ' = G - (S \cup
\{x\}).$ The next claim follows from the definition of $G'$.

\vspace{2mm}

\noindent{\bf Claim B.} {\em The graph $G'$ has no isolated vertices.}



\vspace{3mm}

The next claim follows from the definition of $G'$ and Claim A.

\vspace{2mm}

\noindent{\bf Claim C.} {\em If the maximum degree in $G'$ is at least $2k,$ then $G$ is a Yes-instance of $k$-LCP.}



\vspace{3mm}

The next claim follows from the definition of $G'$ and Rule \ref{r2}.

\vspace{2mm}

\noindent{\bf Claim D.} {\em We have $|V| \le |V'| + k +1.$}


\vspace{3mm}

Observe that if $G'$ was a Yes-instance of $k$-LCP then so would be $G.$ Thus, $G'$ is a No-instance of $k$-LCP.
If the maximum degree in $G'$ is 1, then we may assume that $|V'|< 4k$ as
otherwise by Claim B $G'$ is a matching with at least $2k$ edges and so $(G',k)$ is a
Yes-instance. So, the maximum degree of $G'$ is at least 2.
By Claim C and Lemma \ref{lem:l1}, we may assume that $|V'|
\le 4k + 2k -1 = 6k -1.$ Then by Claim D, $|V| \le 6k -1 + k + 1 = 7k.$
\end{proof}

Using the $7k$ kernel of this section we can get a simple algorithm that tries all 2-colourings of vertices of the kernel.
The running time is $O^*(2^{7k})=O^*(128^k)$. In the next section, we obtain an algorithm of running time $O^*(4^k)$.

\section{Load Coloring Parameterized by Treewidth}\label{sec:tw}

\begin{theorem}\label{thm:dp}
Given a tree decomposition of $G$ of width $t$, we can solve LCP in time
$O(2^{t+1}(k+1)^4n^3)$.
\end{theorem}

\begin{proof}
Let $G=(V,E)$ be graph and let 
$(\mathcal{X},\mathcal{T})$ be a tree decomposition of $G$ of width $t$, where $\mathcal{T} = (I,F)$  and $\mathcal{X} = \{ X_i: i \in I\}$.
By Lemma \ref{lem:niceform}, we may assume that $(\mathcal{X},\mathcal{T})$ is a nice tree decomposition.
 
%
 
  Let $\desc(X_i)$ denote the set of vertices in $V$ which appear in either $X_i$ or a descendant of $X_i$.
 For each $i \in I$, each $S \subseteq X_i$ and each $r,b \in \{0,1,\ldots ,k\}$, define the boolean-valued function $F(X_i,S,r,b)$ to be true if there exists a $2$-coloring $f: \desc(X_i) \rightarrow \{\red, \blue\}$ such that $f^{-1}(\red) \cap X_i = S$ and there are at least $r$ red edges and at least $b$ blue edges in $G[\desc(X_i)]$. We will say such an $f$ \emph{satisfies} $F(X_i,S,r,b)$.
 
 Let $X_0$ denote the bag which is the root of ${\cal T}$, and 
 observe that $G$ is a {\sc Yes}-instance if and only if $F(X_0, S, k,k)$ is true for some $S \subseteq X_0$.
 We now show how to calculate $F(X_i,S,r,b)$ for each $X_i,S,r$ and $b$.
 Assume we have already  calculated  $F(X_j,S',r',b')$ for all descendants $j$ of $i$ and all values of $S',r',b'$.
Our calculation of $F(X_i,S,r,b)$ depends on whether $i$ is a Leaf, Introduce, Forget or Join node.

%
%

 {\bf $i$ is a Leaf node:}
 As $\desc(X_i)=X_i$ there is exactly one $2$-coloring $f: \desc(X_i) \rightarrow \{\red, \blue\}$ such that $f^{-1}(\red) \cap X_i = S$. It is sufficient to set $F(X_i,S,r,b)$ to be true if and only if this coloring gives at least $r$ red edges and at least $b$ blue edges.
 
 {\bf $i$ is an Introduce node:}
 Let $j$ be the child of $i$ and let $v$ be the vertex such that $X_i \setminus X_j = \{v\}$.
 If $v \in S$, let $r^*$ be the number of neighbors of $v$ in $S$.
 Then for any $2$-coloring on $\desc(X_i)$, the number of red edges in $G[\desc(X_i)]$ is exactly the number of red edges in $G[\desc(X_j)]$ plus $r^*$, and the number of blue edges is the same in $G[\desc(X_i)]$ and $G[\desc(X_j)]$.
 Therefore we may set $F(X_i,S,r,b)$ to be true if and only if $F(X_j,S\setminus \{v\}, \max(r-r^*, 0),b)$ is true.
 Similarly if $v \notin S$, we may set $F(X_j,S,r,b)$ to be true if and only if $F(X_j,S\setminus \{v\}, r, \max(b-b^*,0))$ is true, where $b^*$ is the number of neighbors of $v$ in $X_i \setminus S$.
 
 {\bf $i$ is a Forget node:}
 Let $j$ be the child of $i$ and let $v$ be the vertex such that $X_j \setminus X_i = \{v\}$.
 Observe that $\desc(X_i)=\desc(X_j)$, and so any $2$-coloring of $\desc(X_i)$ has exactly the same number of red edges or blue edges in $G[\desc(X_i)]$ and $G[\desc(X_j)]$.
 Therefore we may set $F(X_i,S,r,b)$ to be true if and only if $F(X_j,S,r,b)$ or $F(X_j,S \cup \{v\},r,b)$ is true.
 
 {\bf $i$ is a Join node:}
 Let $h$ and $j$ be the children of $i$.
 Let $r^*$ be the number of red edges in $X_i$. 
 
 Then observe that if there there is a $2$-coloring on $\desc(X_i)$ consistent with $S$ such that there are $r_h$ red edges in $G[\desc(X_h)]$ and $r_j$ red edges in $G[\desc(X_j)]$, then $r_h,r_j \ge r^*$, and the number of red edges in $G[\desc(X_i)]$ is $r_h + r_j - r^*$.

 Let $r' = \min(r^*, k)$. Then if there are at least $r \le k$ red edges in $G[\desc(X_i)]$, there are at least $r_h$ red edges in $G[\desc(X_h)]$ for some $r_h$ such that $r' \le r_h \le k$, and there are at least $r-r_h + r'$ red edges in $G[\desc(X_j)]$.
 Similarly let $b^*$ be the number of blue edges in $X_i$, and let $b' = \min(b^*,k)$. Then if there are at least $b \le k$ blue edges in $G[\desc(X_i)]$, there are at at least $b_h$ blue edges in $G[\desc(X_j)]$ for some $b_h$ such that $b' \le b_h \le k$, and there are at least $b-b_h + b'$ blue edges in $G[\desc(X_j)]$.
 
%
 
 Therefore, we may set $F(X_i,S,r,b)$ to be true if and only if there exist $r_h,b_h$ such that $r' \le r_h \le k$, $b' \le b_h \le k$, and both $F(X_h, S, r_h, b_h)$ and $F(X_j, S, \max(r-r_h+r', 0), \max(b-b_h+b', 0))$ are true.
 
%
%
%
%

\medskip
 
 It remains to analyse the running time of the algorithm.
 
 We first analyse the running time of calculating a single value  $F(X_i,S,r,b)$ assuming we have already calculated $F(X_j,S',r',b')$ for all descendants $j$ of $i$ and all values of $S',r',b'$.
 In the case of a Leaf node, we can calculate $F(X_i,S,r,b)$ in $O(n+m)$ by checking a single $2$-coloring.
 In the case of an Introduce node, we need to check a single value for the child of $i$, and in the case of a Forget node we need to check two values for the child of $i$. Thus, these cases can be calculated in $O(n+m)$ time.
 Finally, for a Join node, we need to check a value from both children of $i$ for every possible way of choosing $r_h,b_h$ such that $r' \le r_h \le k$ and $b' \le b_h \le k$. There are at most $(k+1)^2$ such choices and so we can calculate $F(X_i,S,r,b)$ in $O((k+1)^2(n+m))$ time.
 
It remains to check how many values need to be calculated.
 As there are at most $O(n)$  bags $X_i$, at most $2^{t+1}$ choices of $S \subseteq X_i$, and at most $k+1$ choices for each of $r$ and $b$, the number of values  $F(X_i,S,r,b)$ we need to calculate is $O(n2^{t+1}(k+1)^2)$.
 As each value can be calculated in polynomial time,  overall we have running time 
 $O(2^{t+1}(k+1)^4n(n+m)) = O(2^{t+1}(k+1)^4n^3).$
 \end{proof}

We will combine Theorem \ref{thm:dp} with the following lemma to obtain Theorem \ref{thm:4k}.

\begin{lemma}
For a graph $G$, in polynomial time, we can either determine that $G$ is a Yes-instance of $k$-LCP, or construct a tree decomposition of $G$ of width at most $2k$.
\end{lemma} 
\begin{proof}
If every component of $G$ has at most $k-1$ edges then we may easily
construct a tree decomposition of $G$ of width at most $k-1$ (as each
component has at most $k$ vertices).

Now assume that $G$ has a component with at least $k$ edges.
By starting with a single vertex in a component with at least $k$ edges,
and adding adjacent vertices one at a time, construct a minimal set of
vertices $X$ such that $G[X]$ is connected and $|E(X)| \ge k$, where $E(X)$ is the set of edges with both end-vertices in $X$.
Let $v$ be the last vertex added to $X$. Then $G[X \setminus \{v\}]$ is
connected and $|E(X \setminus \{v\})| < k$. It follows that $|X \setminus \{v\}|
\le k$ and so $|X|\le k+1$.

Now if $|E(V \setminus X)| \ge k$, then we may obtain a solution for $k$-LCP by
coloring all of $X$ red and all of $V \setminus X$ blue.
Otherwise, we may construct a tree decomposition of $G[V \setminus X]$ of
width at most $k-1$. Now add $X$ to every bag in this tree decomposition.
Observe that the result satisfies the conditions of a tree decomposition
and has width at most $k-1 + |X| \le 2k$.
\end{proof}

\begin{theorem}\label{thm:4k}
There is an algorithm of running time $O^*(4^k)$ to solve $k$-LCP.
\end{theorem}

\section{Open Problems}\label{sec:concl}

Our kernel and fixed-parameter algorithm seem to be close to optimal: we do not believe that $k$-LCP admits $o(k)$-vertex kernel or $2^{o(k)}$ running time algorithm unless the Exponential Time Hypothesis fails. It would be interesting to prove or disprove it. It would also be interesting to obtain a smaller kernel or faster algorithm for $k$-LCP.

\medskip

\paragraph{Acknowledgments} We are grateful to the referees for some suggestions that improved results of our paper.

\urlstyle{rm}

\end{document}